\documentclass[journal]{IEEEtran}

\usepackage{mathtools,amssymb,caption}

\usepackage{tabu}
\usepackage{booktabs}
\usepackage{diagbox}
\usepackage{cite}
\usepackage{url}
\usepackage{amsmath}
\usepackage{amsfonts}
\usepackage{dsfont}
\usepackage{amsthm}
\usepackage[breaklinks=true,colorlinks=true,linkcolor=blue,anchorcolor=blue,citecolor=blue,filecolor=blue,urlcolor=blue]{hyperref}
\usepackage{subcaption}
\usepackage{orcidlink}
\usepackage{blkarray, bigstrut}
\usepackage{enumerate}
\usepackage{makecell}
\usepackage{multirow}
\usepackage{stfloats}
\usepackage{siunitx}
\graphicspath{{fig/}}

\newtheorem{theorem}{Theorem}
\newtheorem{proposition}[theorem]{Proposition}

\begin{document}

\title{A Data-Driven Pool Strategy for Price-Makers Under Imperfect Information}

\author{Kedi~Zheng\orcidlink{0000-0001-5407-2892},~\IEEEmembership{Member,~IEEE,}
	    Hongye~Guo\orcidlink{0000-0002-3741-4288},~\IEEEmembership{Member,~IEEE,}
        and~Qixin~Chen\orcidlink{0000-0002-3733-8641},~\IEEEmembership{Senior~Member,~IEEE}
\thanks{Manuscript received June 26, 2021; revised Jan 12, 2022; accepted April 10, 2022. Paper no. TPWRS-01000-2021. This work was supported in part by the Major Smart Grid Joint Project of National Natural Science Foundation of China and State Grid under Grant U2066205 and in part by Natural and Science Foundation of China under Grant 52107102.  (\textit{Corresponding author: Qixin Chen}).}
\thanks{K. Zheng, H. Guo, and Q. Chen are with the State Key Laboratory of Power Systems, Department of Electrical Engineering, Tsinghua University, Beijing 100084 (e-mail: qxchen@tsinghua.edu.cn).}
\thanks{Digital Object Identifier \href{https://doi.org/10.1109/TPWRS.2022.3167096}{10.1109/TPWRS.2022.3167096}}
}

\markboth{MANUSCRIPT FOR IEEE TRANSACTIONS, 2022}%
{Shell \MakeLowercase{\textit{et al.}}: Bare Demo of IEEEtran.cls for IEEE Journals}
%



\maketitle

\IEEEpubidadjcol

\IEEEpubid{\begin{minipage}{\textwidth}\ \\[12pt] \centering
		© 2023 IEEE.  Personal use of this material is permitted.  Permission from IEEE must be obtained for all other uses, in any current or future media, including reprinting/republishing this material for advertising or promotional purposes, creating new collective works, for resale or redistribution to servers or lists, or reuse of any copyrighted component of this work in other works.
\end{minipage}}

\begin{abstract}
This paper studies the pool strategy for price-makers under imperfect information. In this occasion, market participants cannot obtain essential transmission parameters of the power system. Thus, price-makers should estimate the market results with respect to their offer curves using available historical information. The linear programming model of economic dispatch is analyzed with the theory of rim multi-parametric linear programming (rim-MPLP). The characteristics of system patterns (combinations of status flags for generating units and transmission lines) are revealed. A multi-class classification model based on support vector machine (SVM) is trained to map the offer curves to system patterns, which is then integrated into the decision framework of the price-maker. The performance of the proposed method is validated on the IEEE 30-bus system, Illinois synthetic 200-bus system, and South Carolina synthetic 500-bus system.
\end{abstract}


\IEEEpeerreviewmaketitle

\section*{Nomenclature}

\IEEEpubidadjcol

\addcontentsline{toc}{section}{Nomenclature}

Bold symbols denote vectors or matrices. 

\noindent \textit{Sets and Indices}
\begin{IEEEdescription}
	\item[$ b $] Index of blocks.
	\item[$ \mathfrak{b} $] Feasible space for quadratic form bid.
	\item[$ \mathcal{B} $] Set of binding constraint indices.
	\item[$ \mathcal{C} $] Feasible space for block form bid.
	\item[$ \mathcal{G} $] Set of generators owned by a GenCo.
	\item[$ g $] Index for generators.
	\item[$ i,j $] General Indices.
	\item[$ k $] Index of critical regions.
	\item[$ \mathcal{L} $] Feasible space for load.
	\item[$\mathcal{N}$] Set of nodes.
	\item[$ \mathcal{R} $] Critical Region.
	\item[$ \mathcal{X}_{\mathcal{G}} $] Feasible space for the GenCo's decision variables.
\end{IEEEdescription}

\noindent \textit{Variables, Parameters and Functions}
\begin{IEEEdescription}[\IEEEusemathlabelsep\IEEEsetlabelwidth{$ \boldsymbol{P}_G, \boldsymbol{P}_D $}]
	\item[$ a_i, b_i $] Coefficients of the quadratic cost curve at node $ i $.
	\item[$ B $] The number of blocks.
	\item[$ c_{i,b} $] The block price of node $ i $ and block $ b $.
	\item[$ \mathbb{E}_{\boldsymbol{p}} \left( \cdot \right) $] The expected value under probability $ \boldsymbol{p} $.
	\item[$ f $] The decision value in classification.
	\item[$ F^{+/-}_{j} $] Transmission capacity of line $ j $.
	\item[$ \boldsymbol{G} $] Left-hand side coefficients in OPF.
	\item[$ h_g $] The cost function of generator $ g $. 
	\item[$ \boldsymbol{H} $] The diagonal matrix of $ a_i $.
	\item[$ J $] The number of transmission lines.
	\item[$ \boldsymbol{L} $] Load.
	\item[$ N $] The number of nodes.
	\item[$ \boldsymbol{P}_G $] Generated power. 
	\item[$ \boldsymbol{p} $] Posterior probabilities of critical regions.
	\item[$ \boldsymbol{q} $] Cleared generation volume.
	\item[$ \boldsymbol{Q} $] Ancillary matrix in pairwise probability conversion.
	\item[$ Q^{U/L} $] Maximum/minimum generation volume.
	\item[$ r $] Pairwise class probability.
	\item[$ W, \boldsymbol{S} $] Right-hand side coefficients in OPF.
	\item[$ \boldsymbol{w}, \rho $] Parameters of SVM. 
	\item[$ \boldsymbol{x} $] Decision variables of generators.
	\item[$ \boldsymbol{X} $] Input features of the classifier.
	\item[$ y $] Indicator for the class labels.
	\item[$ z $] Revenue of the GenCo.
	\item[$ \beta_{ij} $] Power transmission distribution factor.
	\item[$ \lambda $] The dual variable of the power balance constraint.
	\item[$ \mu^{+/-}  $] The dual variables of the congestion constraints.
	\item[$ \sigma^{U/L} $] The dual variables of the generation capacity constraints.
	\item[$ \boldsymbol{\pi} $] Nodal prices.
	\item[$ \boldsymbol{\Lambda} $] The general form of dual variables in OPF.
	\item[$ \boldsymbol{\varphi}, \boldsymbol{\psi} $] Linear affine functions.
\end{IEEEdescription}

\section{Introduction}

\IEEEpubidadjcol

\IEEEPARstart{P}{rice-makers} are a type of market participants whose market behaviors have a non-negligible impact on market outcomes. In electricity markets, they usually choose to bid strategically and their behaviors are studied in optimization-based frameworks ever since the beginning of market deregulation~\cite{david2000strategic,LI20114686}. In their strategies, both clearing prices and capacity should be considered as endogenous variables formed by the offer curves~\cite{tao2005strategic}, and their estimation of the potential market structure as well as the supply and demand situation underlies such endogenous formation. The modeling of price-makers' pool strategies is beneficial for both market participants and market organizers, as it helps in market simulation, profitability estimation, and market power evaluation.

Currently, the spot electricity market is usually cleared using economic dispatch (ED) and locational marginal price (LMP) scheme~\cite{Litvinov2010lmp}. The system operator collects the offers and the bids from market participants, constructs the ED optimization problem with some constraints and a welfare maximization objective, and then solves the problem ahead of the trading time. The clearing prices are calculated with the dual variables of the energy balance and security constraints (\textit{a.k.a.} the shadow prices) in the ED problem. To ensure solvability and optimality, the objective function is assumed to be convex (linear or quadratic), and the constraints are assumed to be linear. Since essential transmission parameters are considered to be static during a certain time period, the inputs of the ED problem are mainly load forecasts and market bids and offers. 

Market information is very important for the decision of pool strategies for price-makers. Under perfect information, market participants have access to all market data and transmission parameters. Thus, they can construct the precise mathematical model of the ED problem and integrate it into their decision model to form a bilevel optimization problem~\cite{Hobbs2000strategic}. However, when the information is imperfect, the market participants may not have access to transmission parameters and some market data. Their decision model should rely on the estimated or inferred market model from available market data.

Existing research in the field of price-makers' pool strategies falls into two main categories: residual demand curve (RDC)-based methods~\cite{conejo2002optimal,song2018price,xu2011bidding,portela2017residual} and mathematical program with equilibrium constraints (MPEC)-based methods~\cite{Hobbs2000strategic,ruiz2009pool,zugno2013pool,ruiz2014tutorial,ghamkhari2017strategic}. 

The RDC (\textit{a.k.a.} price-quota curve) of a certain generation company (GenCo) is defined as a decreasing function representing its price-quota dependecy~\cite{conejo2002optimal}, and it can be derived from the aggregate supply curve and the aggregate demand curve in the market pool. In its original form, RDC only takes the market supply and demand situation into consideration and ignores the impact resulting from transmission congestion. 
In~\cite{conejo2002optimal}, RDC was integrated into the constraints of GenCo' decision model. Since RDC is stepwise which makes the optimization problem nonlinear, a coordination-descent technique coupled with mixed-integer linear programming (MILP) techniques was proposed. 
In~\cite{song2018price}, the price-maker bidding problem of a retailer with flexible demands was modeled with RDC and solved with MILP techniques. Xu \textit{et al}. extended RDC to apply in occasions when transmission constraints are binding and proposed the concept of transmission-constrained residual demand derivative (TCRDD)~\cite{xu2011bidding}. TCRDD is regarded as sensitivity analysis of the optimal power flow (OPF)\footnote{Since DC-OPF is widely used in electricity markets, the abbreviation OPF stands for DC-OPF rather than AC-OPF in this paper.} model and calculated \textit{ex post} in OPF solution. Gonz\'{a}lez \textit{et al.} extended RDC to consider complex offering conditions in the Iberian electricity market~\cite{portela2017residual}. 

The MPEC-based methods nest the Karush-Kuhn-Tucker (KKT) conditions of OPF in the GenCo's profit maximization problem. By linearizing the product term using the strong duality condition and the complementarity conditions using the big-M method, the MPEC model can be transformed into an MILP problem~\cite{ruiz2009pool}. 
In~\cite{Hobbs2000strategic}, a penalty interior point algorithm was used to compute a local optimal solution of the MPEC. In~\cite{ruiz2009pool}, uncertainties of consumers' bids and rival producers' offers are modeled in a scenario-based approach in MPEC and solved using MILP. In~\cite{zugno2013pool}, MPEC was used to model the pool strategy of a price-maker wind power producer in the balancing market. 
In~\cite{ye2016mpec}, the impact of energy storage on GenCo's strategic behaviors was modeled with MPEC. The market power of large GenCos in imperfect electricity market can be limited by the energy storage system.
Ruiz \textit{et al.} provided a tutorial review on the application of MPEC-related complementarity models in energy market decision problems~\cite{ruiz2014tutorial}. In~\cite{ghamkhari2017strategic}, the problem of high computational burden in the traditional MILP formulation of MPEC was discussed. A convex relaxation method based on Schmudgen's Positivstellensatz Theorem in semialgebraic
geometry is proposed to obtain efficient solutions. 
In~\cite{sepetanc2020cluster}, the optimal bidding and scheduling of aggregated battery swapping stations was studied using MPEC. 

For the RDC-based methods, the main challenge is the modeling of price-makers' influence on the binding constraints of OPF, as RDC and TCRDD are constructed under fixed binding constraints. Although MPEC-based methods have a well-established and solid mathematical foundation and can precisely model price-makers' influence on market results, they are applicable only if perfect market information is available.

Several existing articles~\cite{wen2001optimal,li2005strategic,bompard2021market} discussed the strategic bidding problem of price-maker generators under imperfect information. In~\cite{wen2001optimal}, the bidding problem was formulated in an RDC-like approach where the supply functions of generators were modeled with linear functions and transmission capacity constraints were ignored. It was assumed that each generator did not know the supply functions of its rivals. In~\cite{li2005strategic}, the GenCo's decision problem was modeled using a bilevel problem, and the lower level problem was OPF. Two types of cost coefficients were used to model the unknown supply functions of GenCo's rivals. In~\cite{bompard2021market}, the market equilibrium for a sequential game of a bilateral energy market was analyzed assuming each player had no access to all action of other players that moved before him at every state of the game. 

It is shown that existing discussion on bidding strategy under imperfect information mainly focuses on the modeling of rivals' unknown supply functions, while bidding strategy in absence of transmission topology and parameters is still a problem.

Thus, this paper tries to propose a data-driven pool strategy for price-makers under imperfect information, which makes it more applicable in real markets. The OPF problem is analyzed and modeled using multi-parametric linear programming (MPLP) theory (rim-MPLP, more specifically). The impact of GenCo' bids on the binding constraints is learned using a multi-class support vector machine (SVM) classifier. The classifier is then integrated into the GenCo's decision framework which can be solved using gradient descent methods. The proposed method solves the decision problem in an approach similar to the MPEC-based methods. The major difference is that the proposed method uses a data-driven model instead of the precise OPF in MPEC to model the market outcomes.

The idea of using MPLP theory in the bilevel market decision problem is inspired by reference~\cite{faisca2007parametric,bylling2020parametric}. In~\cite{faisca2007parametric}, it is shown that various classes of bilevel problems can be transformed into a set of independent single level problems using MPLP theory. Such transformation helps in obtaining global optimality. 
In~\cite{bylling2020parametric}, linearly constrained bilevel programming problems with lower-level primal and dual optimal solutions in the upper-level objective function (BPP-Ds) are specifically considered. When the parametric function of the lower level problem can be obtained, BPP-Ds can be solved efficiently. Cases of generation and transmission investment are also performed. 
However, obtaining the parametric functions of the lower level problem also requires a precise OPF model, which limits the direct application of MPLP under imperfect information. 

An important concept in MPLP theory is the critical region (CR) which is also referred to as the system pattern or system pattern region (SPR) when analyzing the OPF problem. Each CR or system pattern corresponds to a unique set of binding constraints in OPF, \textit{i.e.}, combination of status flags for generating units and transmission lines. The shadow prices and optimal solutions of OPF are linear-affine functions within each CR. More theoretical analysis on CR will be provided in Section~\ref{sec:mplp}. Readers can also refer to~\cite{zhou2011short,geng2017learning,zheng2021unsupervised}.

Compared with the aforementioned research, this paper makes the following contributions:
\begin{enumerate}
	\item Electricity markets' basic OPF models are analyzed using rim-MPLP with generation bids and nodal load as parameters of the linear programming problem. The characteristics of critical regions and parametric functions within each CR are revealed. 
	\item An integratable model based on machine learning (specifically a multi-class classifier) is proposed to learn and utilize the probabilistic map between the price-maker's bid and system patterns from available market data. The model can be solved using gradient descent methods and can help the price-maker decide its optimal bidding curve under imperfect information.
	\item Numerical experiments based on the IEEE 30-bus system, Illinois synthetic 200-bus system, and South Carolina synthetic 500-bus system are conducted to demonstrate the performance of the proposed method.
\end{enumerate}

The rest of this paper is structured as follows. Section~\ref{sec:mplp} presents theoretical analysis on OPF with rim MPLP. Section~\ref{sec:learning} details the learning of system patterns using multi-class classification models. Section~\ref{sec:methodology} gives the formulation of the price-maker's optimization problem and its solution method. The case study is conducted in Section~\ref{sec:case}. Finally, Section~\ref{sec:conclusion} draws the conclusion.


\section{Rim-MPLP Analysis for OPF}
\label{sec:mplp}

We consider two typical ways of parameterizing the bidding curve of a generator: the block form and the quadratic form. The OPF problem\footnote{The unit for capacity is assumed to be MW, and the unit for payment and cost is \$, by default.} for the block form can be formulated as follows:

\begin{equation} 
\label{equ:OPF-block}
\begin{aligned}
& \min_{\boldsymbol{P}_G} \sum_{i\in \mathcal{N}} \sum_{b = 1}^{B} c_{i,b} P_{G,i,b}  &  & \\
\mathrm{s.t.} & \sum_{i\in \mathcal{N}} \sum_{b=1}^{B} P_{G,i,b} = \sum_{i\in \mathcal{N}} L_{i} & :\lambda \\  
\\
& \sum_{i\in \mathcal{N}} \beta_{ij} (\sum_{b=1}^{B} P_{G,i,b} - L_{i}) \le F_j^+ & :\mu_j^+ \\
& -\sum_{i\in \mathcal{N}} \beta_{ij} (\sum_{b=1}^{B} P_{G,i,b} - L_{i}) \le F_j^-  & :\mu_j^- \\
& \mathrm{for }\ j = 1,\cdots, J,  \ \mathrm{ and}   & \\ \\
& P_{G,i,b} \le Q_{i,b}^U  &  :\sigma_{i,b}^U \\ 
& -P_{G,i,b} \le -Q_{i,b}^L  & :\sigma_{i,b}^L \\  
& \mathrm{for }\ i \in \mathcal{N}, \ b=1,\cdots,B     & \\
\end{aligned}
\end{equation}
$ \mathcal{N} $ denotes the set for nodes, and $ B $ denotes the number of blocks. For each $ i $, $ c_{i,b} $ is the block price submitted by the generator and satisfies $ c_{i,1} < c_{i,2} < \cdots < c_{i,B} $. $ P_{G,i,b} $ is the generated power for node $ i $ and block $ b $. It is assumed that each node has at most one generator. $ L_i $ is the nodal load, and $ F_j^{+/-} $ denotes the transmission capacity. $ J $ is the number of lines, and $ \beta_{ij} $ is the power transmission distribution factor. $ Q^{U/L} $ denotes the maximum/minimum generation volume. $ \lambda $, $ \mu $, and $ \sigma $ are dual variables. 

Also, the quadratic form is formulated as:

\begin{equation} \label{equ:OPF-quadratic}
\begin{aligned}
& \min_{\boldsymbol{P}_G} \sum_{i\in \mathcal{N}} \frac{1}{2} a_{i} P_{G,i}^2 + b_i P_{G,i}  &  & \\
\mathrm{s.t.} & \sum_{i\in \mathcal{N}} P_{G,i} = \sum_{i\in \mathcal{N}} L_{i} & :\lambda \\  
\\
& \sum_{i\in \mathcal{N}} \beta_{ij} (P_{G,i} - L_{i}) \le F_j^+ & :\mu_j^+ \\
& -\sum_{i\in \mathcal{N}} \beta_{ij} (P_{G,i} - L_{i}) \le F_j^-  & :\mu_j^- \\
& \mathrm{for }\ j = 1,\cdots, J,  \ \mathrm{ and}   & \\ \\
& P_{G,i} \le Q_{i}^U  &  :\sigma_i^U \\ 
& -P_{G,i} \le -Q_{i}^L  & :\sigma_i^L \\  
& \mathrm{for }\ i \in \mathcal{N}     & \\
\end{aligned}
\end{equation}
where $ a_i $ and $ b_i $ are the coefficients of the quadratic cost curve submitted by the generator at node $ i $. 

To summarize, the OPF problems in (\ref{equ:OPF-block}) and (\ref{equ:OPF-quadratic}) can be transformed into the following compact form:
\begin{equation}
\label{equ:mplp}
\begin{aligned}
& \min_{\boldsymbol{P}_G} \boldsymbol{c}^\top \boldsymbol{P}_G & \\ 
\mathrm{s.t.} & \ \boldsymbol{G}_1 \boldsymbol{P}_G = W_1 + \boldsymbol{S}_1 \boldsymbol{L} & : \Lambda_1 \\
\\
& \boldsymbol{G}_i \boldsymbol{P}_G \le W_i + \boldsymbol{S}_i \boldsymbol{L}  & : \Lambda_i \\
& \mathrm{for } \ i = 2, \cdots, 1+2J+2N\times B & \\
\end{aligned}
\end{equation}
for the block form, where $ N = |\mathcal{N}| $. And:
\begin{equation}
\label{equ:mpqp}
\begin{aligned}
& \min_{\boldsymbol{P}_G} \frac{1}{2} \boldsymbol{P}_G^\top \boldsymbol{H} \boldsymbol{P}_G +  \boldsymbol{b}^\top \boldsymbol{P}_G & \\ 
\mathrm{s.t.} & \ \boldsymbol{G}_1 \boldsymbol{P}_G = W_1 + \boldsymbol{S}_1 \boldsymbol{L} & : \Lambda_1 \\
\\
& \boldsymbol{G}_i \boldsymbol{P}_G \le W_i + \boldsymbol{S}_i \boldsymbol{L}  & : \Lambda_i \\
& \mathrm{for } \ i = 2, \cdots, 1+2J+2N & \\
\end{aligned}
\end{equation}
for the quadratic form. In a typical bidding problem of (\ref{equ:OPF-block}), $ c_{i,b} $ and $ Q_{i,b}^U $ are decided by generator $ i $, and $ Q_{i,b}^{L} = 0 $. In (\ref{equ:OPF-quadratic}), $ a_i $ is usually fixed, and $ b_i $ is decided by generator $ i $. 

Assume that $ \beta_{ij} $ and $ F_j^{+/-} $ are fixed, so the parameters in~(\ref{equ:mplp}) and (\ref{equ:mpqp}) are limited to $ \boldsymbol{c} $ and $ \boldsymbol{b} $ in the objective function coefficients (OFC) as well as $ W_i $ and $ \boldsymbol{L} $ in the right-hand side (RHS). The rim-MPLP is defined as a parametric linear programming problem with parameters in both OFC and RHS (Chapter 7 of~\cite{gal2010postoptimal}). When $ Q_{i,b}^U $ is fixed\footnote{It is easy to consider the variation of $ Q_{i,b}^U $ by regarding it as part of the RHS vector $ \boldsymbol{L} $ in (\ref{equ:mplp}) and (\ref{equ:mpqp}).} (and $ W_i $ is consequently fixed), we have the following proposition for the block form. 

\begin{proposition} \label{proposition:rim-mplp}
	The feasible parameter space $ \mathcal{C}\times \mathcal{L} $ for~(\ref{equ:mplp}) is covered by convex polytopes, \textit{a.k.a.} critical regions $ \mathcal{R}_k $, $ k=1,2,\cdots,K $. (i) Within each critical region, the system primal and dual optimal variable solutions can be expressed as \textbf{linear-affine} functions of the parameter vector $ \left(\boldsymbol{c},\boldsymbol{L}\right) $ if the problem is non-degenerate. (ii) The interior of each critical region corresponds to a unique set of binding constraints (system pattern). 
\end{proposition}

\begin{proof}
	The KKT first order necessary conditions for~(\ref{equ:mplp}) can be expressed as follows:
	\begin{equation} \label{equ:kkt-mplp}
	\begin{aligned}
	\boldsymbol{c} + \begin{bmatrix}
	\boldsymbol{G}_1 \\ 
	\vdots \\
	\boldsymbol{G}_{1+2J+2N\times B}
	\end{bmatrix}^\top 
	\begin{bmatrix}
	\Lambda_1\\ 
	\vdots \\
	\Lambda_{1+2J+2N\times B} 
	\end{bmatrix} & = 0    \\
	\boldsymbol{G}_1 \boldsymbol{P}_{G} - W_1 - \boldsymbol{S}_1 \boldsymbol{L} & = 0 \\
	\\
	\Lambda_{i} (\boldsymbol{G}_i \boldsymbol{P}_G - W_i - \boldsymbol{S}_i \boldsymbol{L}) & = 0 \\
	\Lambda_{i} & \ge 0 \\ 
	\boldsymbol{G}_i \boldsymbol{P}_G - W_i - \boldsymbol{S}_i \boldsymbol{L}  & \le 0 \\
	\mathrm{for }\ i = 2,\cdots, 1+2J+2N\times B .
	\end{aligned}
	\end{equation}

	Let $ \mathcal{B} $ denote the set of indices corresponding to the binding equality and inequality constraints in~(\ref{equ:kkt-mplp}). Let $ \boldsymbol{G}_\mathcal{B} $, $ \boldsymbol{W}_\mathcal{B} $, and $ \boldsymbol{S}_\mathcal{B} $ represent the matrices corresponding to $ \mathcal{B} $, and $ \boldsymbol{\Lambda}_\mathcal{B} $ denote the dual variables. Then we have:
	\begin{equation} \label{equ:kkt-mplp-b}
	\begin{gathered}
	\boldsymbol{c} + \boldsymbol{G}_\mathcal{B}^\top  \boldsymbol{\Lambda}_\mathcal{B} = 0 \\
	\boldsymbol{G}_\mathcal{B} \boldsymbol{P}_G - \boldsymbol{W}_\mathcal{B} - \boldsymbol{S}_\mathcal{B} \boldsymbol{L} = 0 
	\end{gathered}
	\end{equation}
	There is a condition for such situation when these binding constraints are linearly independent, named \textit{linear independence constraint qualification (LICQ)}. LICQ holds if $ \boldsymbol{G}_\mathcal{B} $ has full row rank. In the case of OPF, LICQ naturally holds, which is proved in~\cite{zhou2011short}. 
	Then, $ \boldsymbol{G}_\mathcal{B} \boldsymbol{G}_\mathcal{B}^\top $ is invertible, and:
	\begin{equation} \label{equ:mplp-lambda}
	\left(\boldsymbol{\Lambda}^{\mathcal{B}}\right)^* = - \left( \boldsymbol{G}_\mathcal{B} \boldsymbol{G}_\mathcal{B}^\top \right)^{-1} \boldsymbol{G}_\mathcal{B} \boldsymbol{c}
	\end{equation}
	where the symbol $ (\cdot)^* $ indicates the optimal value. We can see that the dual optimal value is a linear-affine function of $ \boldsymbol{c} $. 
	
	The dimension of $ \boldsymbol{G}_\mathcal{B} $ is $ |\mathcal{B}| \times NB $. If $ |\mathcal{B}| < NB $, then there are multiple $ \boldsymbol{P}_G $ satisfying (\ref{equ:kkt-mplp-b}), implying that the problem is degenerate. If the problem is non-degenerate, then $ \boldsymbol{G}_\mathcal{B} $ has full column rank. Since $ \boldsymbol{G}_\mathcal{B} $ has both full row rank and column rank, we have $ |\mathcal{B}| = NB $ and $ \boldsymbol{G}_\mathcal{B} $ is invertible. So the optimal primal value:
	\begin{equation} \label{equ:mplp-pg}
	\left( \boldsymbol{P}_G \right)^* = \boldsymbol{G}_\mathcal{B}^{-1} \left( \boldsymbol{W}_\mathcal{B} + \boldsymbol{S}_\mathcal{B} \boldsymbol{L} \right)
	\end{equation} 
	which is a linear-affine function of $ \boldsymbol{L} $. 
	
	By substituting (\ref{equ:mplp-lambda}) and (\ref{equ:mplp-pg}) into the inequalities in~(\ref{equ:kkt-mplp}), we have:
	\begin{equation}
	\begin{gathered}
	\left\{- \left( \boldsymbol{G}_\mathcal{B} \boldsymbol{G}_\mathcal{B}^\top \right)^{-1} \boldsymbol{G}_i \boldsymbol{c}  \right\}_i  \ge 0 \\
	\boldsymbol{G}_i  \boldsymbol{G}_\mathcal{B}^{-1} \left( \boldsymbol{W}_\mathcal{B} + \boldsymbol{S}_\mathcal{B} \boldsymbol{L} \right) - W_i -\boldsymbol{S}_i \boldsymbol{L} \le 0,  \\
	\mathrm{for } \ i \in \mathcal{B} / \{1\}
	\end{gathered}
	\end{equation}
	which is a convex polytope in $ \mathcal{C}\times \mathcal{L} $ corresponding to the set of binding constraints $ \mathcal{B} $. 
\end{proof}

Similarly, the following proposition holds for the quadratic form.

\begin{proposition} \label{proposition:rim-mpqp}
	The feasible parameter space $ \mathfrak{b}\times \mathcal{L} $ for~(\ref{equ:mpqp}) is covered by convex polytopes, \textit{a.k.a.} critical regions $ \mathcal{R}_k $, $ k=1,2,\cdots,K $. (i) Within each critical region, the system primal and dual optimal variable solutions can be expressed as \textbf{linear-affine} functions of the parameter vector $ \left(\boldsymbol{b},\boldsymbol{L}\right) $ if the problem is non-degenerate. (ii) The interior of each critical region corresponds to a unique set of binding constraints (system pattern). T
\end{proposition}

\begin{proof}
	The KKT first order necessary conditions for~(\ref{equ:mpqp}) can be expressed as follows:
	\begin{equation} \label{equ:kkt-mpqp}
	\begin{aligned}
	\boldsymbol{H}\boldsymbol{P}_G + \boldsymbol{b} + \begin{bmatrix}
	\boldsymbol{G}_1 \\ 
	\vdots \\
	\boldsymbol{G}_{1+2J+2N}
	\end{bmatrix}^\top 
	\begin{bmatrix}
	\Lambda_1\\ 
	\vdots \\
	\Lambda_{1+2J+2N} 
	\end{bmatrix} & = 0    \\
	\boldsymbol{G}_1 \boldsymbol{P}_{G} - W_1 - \boldsymbol{S}_1 \boldsymbol{L} & = 0 \\
	\\
	\Lambda_{i} (\boldsymbol{G}_i \boldsymbol{P}_G - W_i - \boldsymbol{S}_i \boldsymbol{L}) & = 0 \\
	\Lambda_{i} & \ge 0 \\ 
	\boldsymbol{G}_i \boldsymbol{P}_G - W_i - \boldsymbol{S}_i \boldsymbol{L}  & \le 0 \\
	\mathrm{for }\ i = 2,\cdots, 1+2J+2N.
	\end{aligned}
	\end{equation}
	
	Let $ \mathcal{B} $ denote the set of indices corresponding to the binding equality and inequality constraints. We have:
	\begin{equation} \label{equ:kkt-mpqp-b}
	\begin{aligned}
	\boldsymbol{H}\boldsymbol{P}_G + \boldsymbol{b} + \boldsymbol{G}_\mathcal{B}^\top  \boldsymbol{\Lambda}_\mathcal{B} & = 0 \\
	\boldsymbol{G}_\mathcal{B} \boldsymbol{P}_G - \boldsymbol{W}_\mathcal{B} - \boldsymbol{S}_\mathcal{B} \boldsymbol{L} & = 0 
	\end{aligned}
	\end{equation}
	
	Since $ \boldsymbol{H} $ is positive-definite and $ \boldsymbol{G}_\mathcal{B} $ has full row rank, $ \boldsymbol{G}_\mathcal{B} \boldsymbol{H}^{-1} \boldsymbol{G}_\mathcal{B} $ is invertible. The optimal variables can be solved as:
	\begin{equation} \label{equ:mpqp-optimal}
	\begin{aligned}
	\left(\boldsymbol{\Lambda}^{\mathcal{B}}\right)^* = & - \left( \boldsymbol{G}_\mathcal{B} \boldsymbol{H}^{-1} \boldsymbol{G}_\mathcal{B}^\top \right)^{-1} \left( \boldsymbol{G}_\mathcal{B} \boldsymbol{H}^{-1} \boldsymbol{b} + \right. \\ 
	& \left. \boldsymbol{W}_\mathcal{B} + \boldsymbol{S}_\mathcal{B} \boldsymbol{L} \right)  \\ 
	\left( \boldsymbol{P}_G \right)^* =  & - \boldsymbol{H}^{-1} \boldsymbol{b} + \boldsymbol{H}^{-1} \boldsymbol{G}_{\mathcal{B}}^\top \left( \boldsymbol{G}_\mathcal{B} \boldsymbol{H}^{-1} \boldsymbol{G}_\mathcal{B}^\top \right)^{-1} \\
	& \left( \boldsymbol{G}_\mathcal{B} \boldsymbol{H}^{-1} \boldsymbol{b} + \boldsymbol{W}_\mathcal{B} + \boldsymbol{S}_\mathcal{B} \boldsymbol{L} \right)
	\end{aligned}
	\end{equation}
	which is linear-affine functions on $ \mathfrak{b} \times \mathcal{L} $. 
	
	By substituting (\ref{equ:mpqp-optimal}) into the inequalities in (\ref{equ:kkt-mpqp}), we have:
	\begin{equation}
	\begin{gathered}
	\left\{- \left( \boldsymbol{G}_\mathcal{B} \boldsymbol{H}^{-1} \boldsymbol{G}_\mathcal{B}^\top \right)^{-1} \left( \boldsymbol{G}_\mathcal{B} \boldsymbol{H}^{-1} \boldsymbol{b} + \boldsymbol{W}_\mathcal{B} + \boldsymbol{S}_\mathcal{B} \boldsymbol{L} \right) \right\}_i \ge 0  \\ 
	\\
	W_i + \boldsymbol{S}_i \boldsymbol{L} \ge \boldsymbol{G}_i \left(  - \boldsymbol{H}^{-1} \boldsymbol{b} + 
	\boldsymbol{H}^{-1} \boldsymbol{G}_{\mathcal{B}}^\top 
	\right. \\ \left.
	\left( \boldsymbol{G}_\mathcal{B} \boldsymbol{H}^{-1} \boldsymbol{G}_\mathcal{B}^\top \right)^{-1}   
	\left( \boldsymbol{G}_\mathcal{B} \boldsymbol{H}^{-1} \boldsymbol{b} + \boldsymbol{W}_\mathcal{B} + \boldsymbol{S}_\mathcal{B} \boldsymbol{L} \right)  \right) \\
	\mathrm{for } \ i \in \mathcal{B} / \{1\}
	\end{gathered}
	\end{equation}
	which is a convex polytope on $ \mathfrak{b} \times \mathcal{L} $ corresponding to the set of binding constraints $ \mathcal{B} $. 
\end{proof}

Propositions~\ref{proposition:rim-mplp} and~\ref{proposition:rim-mpqp} set up the foundation for learning the parametric functions of $ \boldsymbol{\Gamma}^* $ and $ \boldsymbol{P}_G^* $ under imperfect information. Although market participants may not have access to transmission parameters such as $ \beta_{ij} $ and $ F_{j}^{+/-} $, the map between the parameters of $ (\boldsymbol{c},\boldsymbol{L}) $ or $ (\boldsymbol{b},\boldsymbol{L}) $ and their attribution of CR can be learned using machine learning methods. Also, the linear-affine coefficients and intercepts for each CR can be estimated using linear regression. 

Fig.~\ref{fig:example} gives an example of CRs using the simple three-bus system from~\cite{geng2017learning,zheng2021unsupervised}. The block form is used with $ B = 1 $. $ c_{1,1} $ and $ c_{2,1} $ as well as $ L_2 $ and $ L_3 $ are considered as parameters. To visualize the CRs, randomly generated points of $ (c_{1,1},c_{2,1},L_2,L_3) $ are plotted using multidimensional scaling (MDS) in a three-dimensional space. The CR borders are marked using transparent colored faces. 

The main results in this section are derived from lossless OPF models of (\ref{equ:OPF-block}) and (\ref{equ:OPF-quadratic}). When extending the results to lossy OPF models~\cite{litvinov2004marginal,hu2010iterative}, the characteristics of CR still hold as long as the loss sensitivity vector and loss distribution factor vector are fixed. If the two loss-related vectors
are updated iteratively according to real-time market conditions, there will be overlaps between different CRs. However, such overlaps do not affect the proposed learning method in Section~\ref{sec:learning}. The case study on this will also be given and discussed in Section~\ref{sec:case}. 

\begin{figure}[!t]
	\centering
	\begin{subfigure}[t]{1.0\linewidth}
		\centering\includegraphics[width=1.0\linewidth]{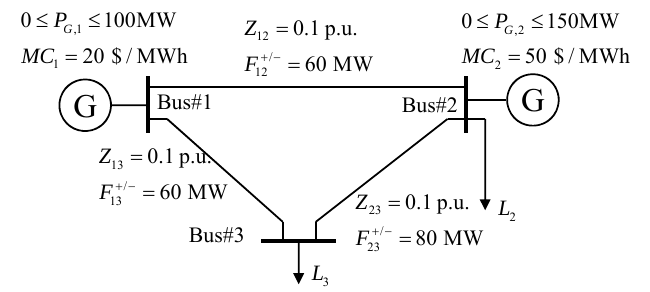}
		\caption{The 3-bus system in~\cite{geng2017learning,zheng2021unsupervised}.}
		\label{subfig:three-bus}
	\end{subfigure}
	\\
	\begin{subfigure}[t]{1.0\linewidth}
		\centering\includegraphics[width=1.0\linewidth]{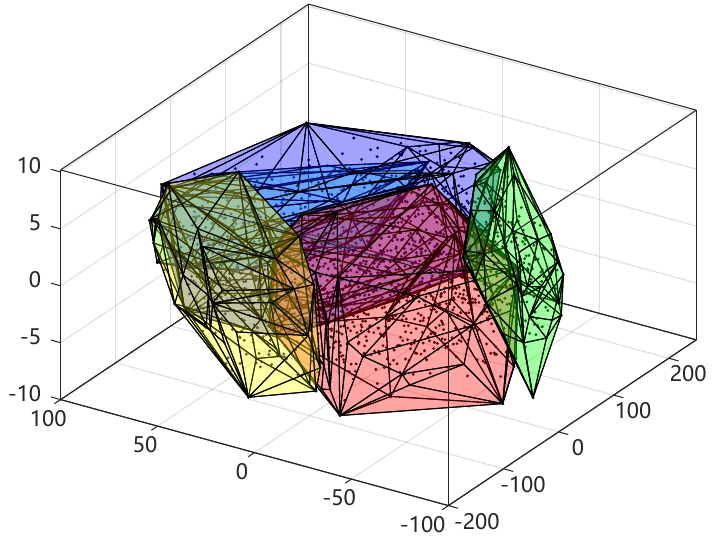}
		\caption{The MDS plot of the four CRs. Axes: reduced dimensional feature space. Dots: feature points. Transparent polytopes: CRs.}
		\label{subfig:CRs}
	\end{subfigure}
	\caption{An example of CRs in a small-sized OPF case.}
	\label{fig:example}
\end{figure}

\section{Learning System Patterns}
\label{sec:learning}

According to the hyperplane separation theorem in geometry (Section 2.5.1~\cite{boyd2004convex}), there exists a separating hyperplane for each pair of two CRs in the high-dimensional Euclidean space, which makes it possible to learn the system patterns from parameters using machine learning-based classifiers in a one-vs-one (OVO) approach. Among these classifiers, we choose SVM because it can fit the hyperplane in a natural way and allow tolerance for CR overlapping by incorporating soft margins. In addition to learning the CRs, the linear-affine coefficients in the parameter functions should also be estimated. 

\subsection{Learning Probabilities of CRs}

For consistency, we use the notation $ \boldsymbol{x} $ to represent the decision variables of the generators, \textit{i.e.}, $ \boldsymbol{c} $ and $ \boldsymbol{Q}^{U} $ in block form and $ \boldsymbol{b} $ in quadratic form. 
Let $ p_k(\boldsymbol{x},\boldsymbol{L}) $ denote the posterior probability of the $ k $-th CR $ \Pr(\mathcal{R}_k | \boldsymbol{x}, \boldsymbol{L}) $. Let $ \boldsymbol{X} = [\boldsymbol{x}; \boldsymbol{L}]  $. 
The input of the OVO classifier is the feature vector $ \boldsymbol{X} $ or part of vector $ \boldsymbol{X} $ (see the discussion in Section~\ref{subsec:information})
The output is a vector $ \boldsymbol{p} = [p_k(\boldsymbol{x},\boldsymbol{L})]_k $ representing the estimated posterior probabilities of CR belonging.

In machine learning, the OVO estimation of multi-class posterior probabilities contains the following steps:
\begin{enumerate}
	\item OVO classification. Obtain $ K \choose 2 $ binary classifers. Each classifier between classes $ i $ and $ j $ has a decision value $ f_{ij} = f_{ij}(\boldsymbol{X}) $. Here, we take SVM as an example, but note that the engine of binary classifiers is not limited to SVM. Let $ t $ denote the index of observations of $ \boldsymbol{X} $ whose CR falls in $ \mathcal{R}_i $ or $ \mathcal{R}_j $. Let $ y^{(t)} $ be the indicator for $ \boldsymbol{X}^{(t)} $ so that $ y^{(t)} = -1 $ if $ \boldsymbol{X}^{(t)} $'s CR is $ \mathcal{R}_i $, and $ y^{(t)} = 1 $ if $ \boldsymbol{X}^{(t)} $'s CR is $ \mathcal{R}_j $. The following convex optimization problem is solved:
	\begin{equation}
	\begin{aligned}
	\min_{\boldsymbol{w}_{ij}, \rho_{ij}, \boldsymbol{s}} \ \ & \frac{1}{2} \boldsymbol{w}_{ij}^\top \boldsymbol{w}_{ij} + C ||\boldsymbol{s}||_{1} \\
	\mathrm{s.t.} \ \ & y^{(t)} (\boldsymbol{w}_{ij}^\top \boldsymbol{X}^{(t)}- \rho_{ij}) \ge 1- s^{(t)} \\
	& s^{(t)} \ge 0, y^{(t)} \in \{-1,1\}
	\end{aligned}
	\end{equation}
	where $ C $ is the tolerance for CR overlapping, and $ ||\cdot||_1 $ denotes the $ \ell_1$ norm. The decision value is calculated:
	\begin{equation}
	f_{ij} = \boldsymbol{w}_{ij}^\top \boldsymbol{X} - \rho_{ij}
	\end{equation}
	\item OVO probability calibration with Platt's algorithm~\cite{platt1999probabilistic} (a.k.a. Platt scaling) which is a widely-used method in probability estimation for binary classifiers in well-known software such as scikit-learn and LIBSVM~\cite{scikit-learn,libsvm}. The algorithm estimates parameters of $ A_{ij} $ and $ B_{ij} $ for the binary classifier between classes $ i $ and $ j $ ($ i<j $). The pairwise class probabilities $ r_{ij} $ are obtained by:
	\begin{equation}
	r_{ij} = \begin{cases}
	\frac{1}{ 1 + e^{A_{ij} f_{ij} + B_{ij} } } & i < j \\
	1- r_{ji} & i > j 
	\end{cases}
	\end{equation}
	\item Converting the pairwise probabilities $ r_{ij} $ to $ p_k $. There are several ways of doing this~\cite{wu2004probability}. We use the second approach proposed by Wu \textit{et al.} which is also the one implemented in LIBSVM. First, obtain the matrix $ \boldsymbol{Q} $ by:
	\begin{equation} \label{equ:solve-Q}
	Q_{ij} = \begin{cases}
	\sum_{s: s\neq i} r_{si}^2 & \mathrm{if}\ i = j, \\
	- r_{ji} r_{ij} & \mathrm{if}\ i\neq j . 
	\end{cases}
	\end{equation}
	Then solve the following quadratic problem and obtain $ \boldsymbol{p} $:
	\begin{equation} \label{equ:solve-p}
	\begin{aligned}
	\min_{\boldsymbol{p}}\ & \frac{1}{2} \boldsymbol{p}^\top \boldsymbol{Q} \boldsymbol{p} \\
	\mathrm{s.t.}\ & \boldsymbol{e}^\top \boldsymbol{p} = 1
	\end{aligned}	
	\end{equation}
	which has the following KKT condition:
	\begin{equation} \label{equ:kkt-pQ}
	\begin{bmatrix}
	\boldsymbol{Q} & \boldsymbol{e} \\
	\boldsymbol{e}^\top & 0
	\end{bmatrix}
	\begin{bmatrix}
	\boldsymbol{p} \\
	\alpha
	\end{bmatrix} = 
	\begin{bmatrix}
	\boldsymbol{0} \\
	1
	\end{bmatrix}
	\end{equation}
	where $ \alpha $ is a scalar. 
\end{enumerate}

\subsection{Learning of Parametric Functions}
Let $ \boldsymbol{\pi} $ denote the nodal prices resulting from OPF, which can be obtained from the shadow prices $ \boldsymbol{\Lambda} $~\cite{litvinov2004marginal,zheng2021unsupervised}. Since $ \boldsymbol{\Lambda} $ are linear-affine functions within each CR, $ \boldsymbol{\pi} $ are also linear-affine functions w.r.t. $ \boldsymbol{X} $ in each CR. 
For a system pattern $ k $, let the following function denote $ \boldsymbol{\pi} $ and the cleared generation volume $ \boldsymbol{q} = \boldsymbol{P}_G^* $:
\begin{equation}
\begin{aligned}
\boldsymbol{\pi} = \boldsymbol{\varphi}^{(k)} (\boldsymbol{x}, \boldsymbol{L}) = \boldsymbol{\varphi}_x^{(k)} \boldsymbol{x} + \boldsymbol{\varphi}_L^{(k)} \boldsymbol{L} + \varphi_0^{(k)}  \\
\boldsymbol{q} = \boldsymbol{\psi}^{(k)} (\boldsymbol{x}, \boldsymbol{L}) =\boldsymbol{\psi}_x^{(k)} \boldsymbol{x} + \boldsymbol{\psi}_L^{(k)} \boldsymbol{L} + \psi_0^{(k)} 
\end{aligned}
\end{equation}
The coefficient matrices $ \boldsymbol{\varphi}_x^{(k)} $, $ \boldsymbol{\varphi}_L^{(k)} $, $ \boldsymbol{\psi}_x^{(k)} $, $ \boldsymbol{\psi}_L^{(k)} $ and intercepts $ \varphi_0^{(k)} $, $ \psi_0^{(k)} $ can be estimated using linear regression given market observations $ \boldsymbol{\pi}^{(t)} $ and $ \boldsymbol{X}^{(t)} $ if $ \boldsymbol{X}^{(t)} \in \mathcal{R}_k $.

\section{Optimization Framework}
\label{sec:methodology}

This section introduces the optimal bidding problem for GenCos as price-makers and provides a discussion on the market information level. It also presents the overall framework for the proposed data-driven pool strategy.

\subsection{GenCo's Optimization Problem}
\label{subsec:optimization}

Assume that, $ K $ classes (of system patterns) are considered in the strategic bidding problem. For a GenCo defined by set $ \mathcal{G} \subset \mathcal{N} $ of generators located at $ g \in \mathcal{G} $, the locational marginal price $ \pi_g $ and the cleared volume $ q_g $ determine its income from market settlement. 

For the generator at $ g $, define its cost as $ h_g(q_g) $. Thus, for the GenCo defined by $ \mathcal{G} $, its optimal offering problem is formulated as (with $ \boldsymbol{x}_{\mathcal{N}/\mathcal{G}} $ and $ \boldsymbol{L} $ fixed):
\begin{equation}
\label{equ:max-revenue}
\begin{aligned}
\max_{\boldsymbol{x}_\mathcal{G}} z & =  \mathbb{E}_{\boldsymbol{p}}\left( \boldsymbol{\pi}_\mathcal{G}^\top \boldsymbol{q}_\mathcal{G} - \sum_{g \in \mathcal{G}} h_g\left( q_g \right) \right) \\
=
\sum_{k=1}^{K} p_k \left( \boldsymbol{x},\boldsymbol{L} \right) \cdot & \left[ \left(\boldsymbol{\varphi}^{(k)} (\boldsymbol{x}, \boldsymbol{L}) \right)_{\mathcal{G}} \cdot \left(\boldsymbol{\psi}^{(k)} (\boldsymbol{x}, \boldsymbol{L}) \right)_{\mathcal{G}} \right. \\ & \left. - \sum_{g\in\mathcal{G}} h_g \left( \psi^{(k)}_{g} (\boldsymbol{x}, \boldsymbol{L}) \right) \right] \\ 
\mathrm{s.t.} \ & \boldsymbol{x}_{\mathcal{G}} \in \mathcal{X}_\mathcal{G}
\end{aligned}
\end{equation}
Note that:
\begin{equation}
\begin{aligned} 
& \boldsymbol{\nabla}_{\boldsymbol{x}_{\mathcal{G} } } \left[ \left(\boldsymbol{\varphi}^{(k)} (\boldsymbol{x}, \boldsymbol{L}) \right)_{\mathcal{G}} \left(\boldsymbol{\psi}^{(k)} (\boldsymbol{x}, \boldsymbol{L}) \right)_{\mathcal{G}} \right. \\
& \left.
- \sum_{g\in\mathcal{G}} h_g \left( \psi^{(k)}_{g} (\boldsymbol{x}, \boldsymbol{L}) \right) \right] \\ 
= & \left[ \boldsymbol{\nabla}_{\boldsymbol{x}} (\boldsymbol{\varphi}^{(k)})_{\mathcal{G}} \right]_{\mathcal{G}} \left(\boldsymbol{\psi}^{(k)} \right)_{\mathcal{G}} +  \left[ \boldsymbol{\nabla}_{\boldsymbol{x}} (\boldsymbol{\psi}^{(k)})_{\mathcal{G}} \right]_{\mathcal{G}} \left(\boldsymbol{\varphi}^{(k)}  \right)_{\mathcal{G}} \\ 
& - \sum_{g\in\mathcal{G}} h'_g \left( \nabla_{\boldsymbol{x}} \psi^{(k)}_g \right)_{\mathcal{G}}
\end{aligned}
\end{equation}

\begin{table*}[t!]
	\centering
	\begin{minipage}{1.0\textwidth}
		\begin{equation} \label{equ:gradient-z}
		\begin{aligned}
		\boldsymbol{\nabla}_{\boldsymbol{x}_{\mathcal{G} } }  z  = &  \boldsymbol{\nabla}_{\boldsymbol{x}_{\mathcal{G} } }  \  \boldsymbol{p} \cdot \left[ \left(\boldsymbol{\varphi}^{(k)} (\boldsymbol{x}, \boldsymbol{L}) \right)_{\mathcal{G}} \left(\boldsymbol{\psi}^{(k)} (\boldsymbol{x}, \boldsymbol{L}) \right)_{\mathcal{G}} - \sum_{g\in\mathcal{G}} h_g \left( \psi^{(k)}_{g} (\boldsymbol{x}, \boldsymbol{L}) \right) \right]_{k=1}^{K} 
		= 
		\left(  \boldsymbol{\nabla}_{\boldsymbol{x}} \boldsymbol{p}  \right)_{\mathcal{G}}  \left[ \left(\boldsymbol{\varphi}^{(k)} \right)_{\mathcal{G}} \left(\boldsymbol{\psi}^{(k)} \right)_{\mathcal{G}} - \sum_{g\in\mathcal{G}} h_g \left( \psi^{(k)}_{g} \right) \right]_{k=1}^{K} \\
		& + 
		\left( \left[ \boldsymbol{\nabla}_{\boldsymbol{x}} (\boldsymbol{\varphi}^{(k)})_{\mathcal{G}} \right]_{\mathcal{G}} \left(\boldsymbol{\psi}^{(k)} \right)_{\mathcal{G}} + \left[ \boldsymbol{\nabla}_{\boldsymbol{x}} (\boldsymbol{\psi}^{(k)})_{\mathcal{G}} \right]_{\mathcal{G}}  \left(\boldsymbol{\varphi}^{(k)}  \right)_{\mathcal{G}}   - \sum_{g \in \mathcal{G}} h'_g \left[ \nabla_{\boldsymbol{x}} \psi^{(k)}_g \right]_{\mathcal{G}} \right) \boldsymbol{p}
		\end{aligned}
		\end{equation}
		\medskip
		\hrule
	\end{minipage}
\end{table*}

Thus, the gradient of the GenCo's objective function can be shown as~(\ref{equ:gradient-z}). An important item in (\ref{equ:gradient-z}) is the gradient of the probability output of the multi-class classifier $  \left( \boldsymbol{\nabla}_{\boldsymbol{x}} \boldsymbol{p}  \right)_{\mathcal{G}}  $, which defines the market power of the GenCo on altering the binding constraints of the market under certain operation conditions. The calculation of $ \left( \boldsymbol{\nabla}_{\boldsymbol{x}} \boldsymbol{p} \right)_{\mathcal{G}} $ is given in the Appendix~\ref{appendix}. 

Once the gradient is obtained, the problem can be solved by gradient descent. In this paper, we use fixed learning rate in gradient descent:
\begin{equation}
\label{equ:gradient-descent}
\boldsymbol{x}_{\mathcal{G}}^{(i+1)} := \boldsymbol{x}_{\mathcal{G}}^{(i)} + \eta \cdot \frac{\boldsymbol{\nabla}_{\boldsymbol{x}_\mathcal{G}} z}{|| \boldsymbol{\nabla}_{\boldsymbol{x}_\mathcal{G}} z ||_2 }
\end{equation}
where $ ||\cdot||_2 $ is $ \ell_2 $ norm.

\begin{figure}[t]
	\centering
	\includegraphics[width=0.75\linewidth]{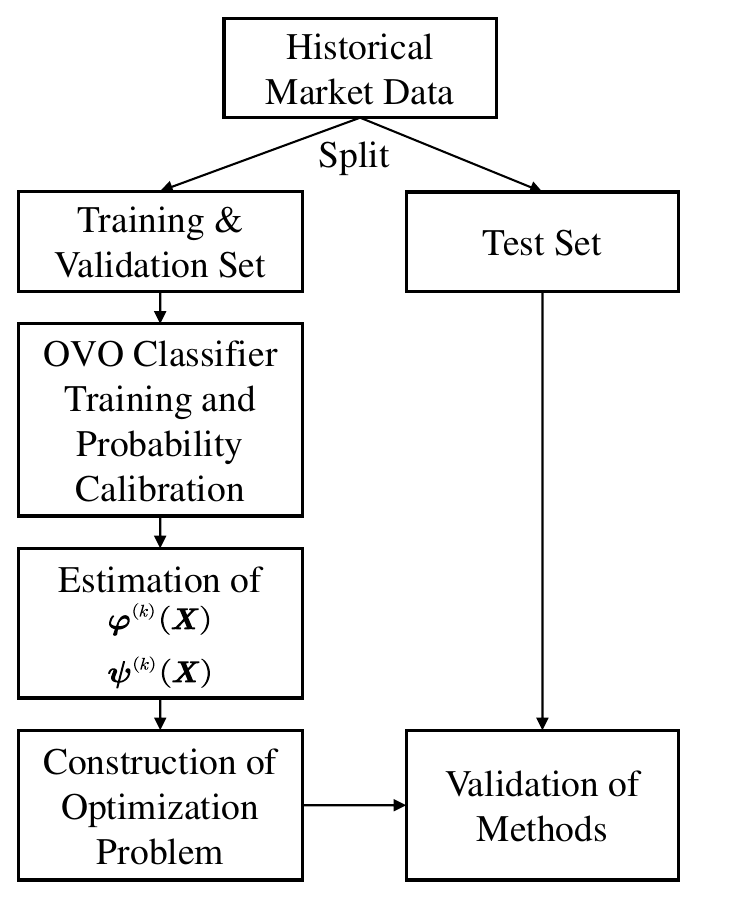}
	\caption{Framework of the proposed method.}
	\label{fig:framework}
\end{figure}

\subsection{Discussion on Available Information}
\label{subsec:information}

The method proposed in Section~\ref{sec:learning} and Section~\ref{subsec:optimization} are based on the assumption that $ \boldsymbol{x}^{(t)} $ and $ \boldsymbol{L}^{(t)} $ are available to the GenCo. However, the level of market information disclosure depends on specific rules that differ in different markets. Thus, we consider the following levels of available information:
\begin{enumerate}[I.]
	\item Perfect information. All market data and the transmission parameters are available. Under perfect information, the GenCo can construct the MPEC model to evaluate its market power and theoretically achieve maximum revenue. 
	\item Imperfect information. The transmission parameters are unavailable. All market data are available. 
	\item Imperfect information. The transmission parameters are unavailable. Data of $ \boldsymbol{L}^{(t)} $ are available. The GenCo only knows its own bidding data. Thus, only part of $ \boldsymbol{x}^{(t)} $ is available. 
	\item Imperfect information. The transmission parameters are unavailable. Only zonal load data (rather than nodal load data) and GenCo's own bidding data are available.
\end{enumerate}

In levels III and IV, the original feature vector $ \boldsymbol{X}^{(t)} $ is unavailable. A new vector can be constructed based on the available information of the GenCo. The learning of system patterns and the GenGo's optimization problem are performed with the new vector. The procedure is unchanged, but the accuracy of learning is likely to be lower since incomplete features result in a higher degree of CR overlapping. In levels III and IV, the new feature vector is also denoted using $ \boldsymbol{X}^{(t)} $ for the rest of this paper.

\subsection{The Overall Framework}

The framework of the proposed method is shown in Fig.~\ref{fig:framework}. 
The whole set of market data is split into the training \& validation set and test set. Since SVM is sensitive to the feature scale, a Min-Max scaler is used to preprocess the features. The classifier and its hyperparameters are trained and tuned on the training \& validation set. Also, the coefficients in the parametric functions $ \boldsymbol{\varphi}^{(k)} $ and $ \boldsymbol{\psi}^{(k)} $ for each system pattern are estimated using linear regression. Then, the optimal bidding problem based on the learned system patterns is constructed and solved, and the test set is used to validate the proposed method.

\subsection{Complexity Analysis}

The training of OVO SVM and its probabilistic calibration are well developed in machine learning software. Once the model is trained, it can be used for a certain time period. The time complexity for solving the optimization problem using gradient descent mainly comes from the calculation of the gradient, specifically $ \left( \boldsymbol{\nabla}_{\boldsymbol{x}} \boldsymbol{p} \right)_{\mathcal{G}} $. The main computational burden comes from the calculation of $ \boldsymbol{M} $ in (\ref{equ:delta_p}) and the updating of $ \boldsymbol{p} $ and $ \boldsymbol{Q} $ in (\ref{equ:solve-p})-(\ref{equ:solve-Q}). The calculation of $ \boldsymbol{M} $ requires a matrix inversion operation on $ \boldsymbol{Q} $ with a time complexity of $ O(K^3) $. The complexity of updating $ \boldsymbol{q} $ is also $ O(K^3) $ using Gaussian elimination on (\ref{equ:kkt-pQ}). The complexity of updating $ \boldsymbol{Q} $ is $ O(N^2) $, where $ N $ defines the length of the feature vector $ \boldsymbol{X} $.

\section{Case Study}
\label{sec:case}
To illustrate the application of the data-driven pool strategy, a case study based on simulated market data is conducted.

We perform the case study using the IEEE 30-bus system, Illinois synthetic 200-bus system, and South Carolina synthetic 500-bus system~\cite{birchfield2017grid} whose data are provided in Matpower~\cite{matpower}. In the 30-bus system, all 6 generators are assumed to be strategic. In the 200-bus system, generators \#26, \#27, \#28, \#29, \#30, and \#37 are assumed to be strategic, while the other 43 generators that have small capacity or low generation cost are assumed to be nonstrategic. In the 500-bus system, 14 generators with generation capacity of more than \SI{100}{MW} and relatively higher generation cost are set as strategic, while other 42 generators are set as nonstrategic. It is assumed that generators can price away freely from their true cost.

All numerical experiments are conducted on an Intel Core i9-10900K@3.70 GHz desktop with MATLAB 2020b. SVM is trained using LIBSVM~\cite{libsvm}.

For each test system, both the block form and the quadratic form of bid curves as well as both lossless and lossy OPF models are considered. For the block form, the number of price blocks $ B $ is set as 5, including 2 nonstrategic blocks (which are set at a sufficiently low price to ensure a minimum income) for the 30-bus system case and 1 nonstrategic block for the 200-bus system and 500-bus system. Each block has the same volume in the basic scenario. For each case (combination of different systems and block/quadratic and lossless/lossy), 8\,760 scenarios are generated by adding 10\% Gaussian deviation on $ \boldsymbol{c}_{i,b} $ and $ \boldsymbol{Q}_{i,b}^{U} $ or $ \boldsymbol{b}_{i} $ of the basic scenario. The nodal load scenarios are generated using the one-year hourly load scenarios of the Illinois synthetic 200-bus system contained in the Matpower scenario file.

The whole dataset is split into two sets: the training \& validation set (80\%) and the test set (20\%). The training and validation procedure is performed in a 5-fold cross validation approach. The best hyperparameter $ C $ in SVM is chosen from $ \{0.1,1,10,10,100,10^3, 10^4\}$. The learning rate in (\ref{equ:gradient-descent}) is fixed at 0.01, and the maximum number of iterations is set at 200.

\subsection{Comparisons}

Different levels of information are tested as mentioned in Section~\ref{subsec:information}. 
Three methods are used as benchmark for comparison. 
The first one is MPEC with information level I. It is solved using the big-M method in MILP with Gurobi. 
The second one is an RDC-like method. Assuming that the congestion status is fixed, the sensitivity of market results with respect to generator bids is estimated using historical data, and is embedded into the bidding problem.
The third method uses the same information as level II but replaces the CR's probability $ p_k(\boldsymbol{x},\boldsymbol{L}) $ in (\ref{equ:max-revenue}) with its empirical value $ \bar{p}_k $. This means that the impact of the GenCo's bid on the market results of system patterns is ignored. For the sake of simplicity, the third method using the empirical probability is further referred to as level V.

\subsection{Classification Results}

\begin{table}[!t]
	\renewcommand{\arraystretch}{1.3}
	\caption{Classification Accuracy (\%) Under Different Case Settings}
	\label{table:classification}
	\centering
	\begin{tabu}{ccccccc}
		\toprule
		\makecell{Case\\size} & \makecell{Bid curve\\form} & \makecell{OPF\\type} & \makecell{Dummy\\1} & \makecell{Dummy\\2} & \makecell{SVM\\Train} & \makecell{SVM\\Test} \\ \midrule
		30 & Block & Lossless & 7.73 & 3.34 & 99.50 & 60.51 \\ 
		30 & Block & Lossy & 8.15 & 3.34 & 95.41 & 61.66 \\ 
		30 & Quadratic & Lossless & 54.53 & 32.86 & 98.48 & 91.70 \\ 
		30 & Quadratic & Lossy & 47.26 & 25.91 & 98.75 & 92.34 \\ \hline
		200 & Block & Lossless & 17.81 & 6.65 & 76.26 & 59.43 \\ 
		200 & Block & Lossy & 26.34 & 10.41 & 61.89 & 55.99 \\ 
		200 & Quadratic & Lossless & 38.27 & 17.12 & 88.46 & 75.02 \\ 
		200 & Quadratic & Lossy & 46.67 & 24.01 & 89.39 & 72.43 \\ \hline
		500 & Block & Lossless & 8.66 & 3.46 & 98.56 & 46.57 \\ 
		500 & Block & Lossy & 9.97 & 3.61 & 99.05 & 42.19 \\ 
		500 & Quadratic & Lossless & 14.66 & 7.09 & 94.39 & 69.52 \\ 
		500 & Quadratic & Lossy & 14.98 & 6.48 & 92.86 & 66.87 \\ \bottomrule
	\end{tabu}
\end{table}

\begin{figure}[!t]
	\centering
	\includegraphics[width=1.0\linewidth]{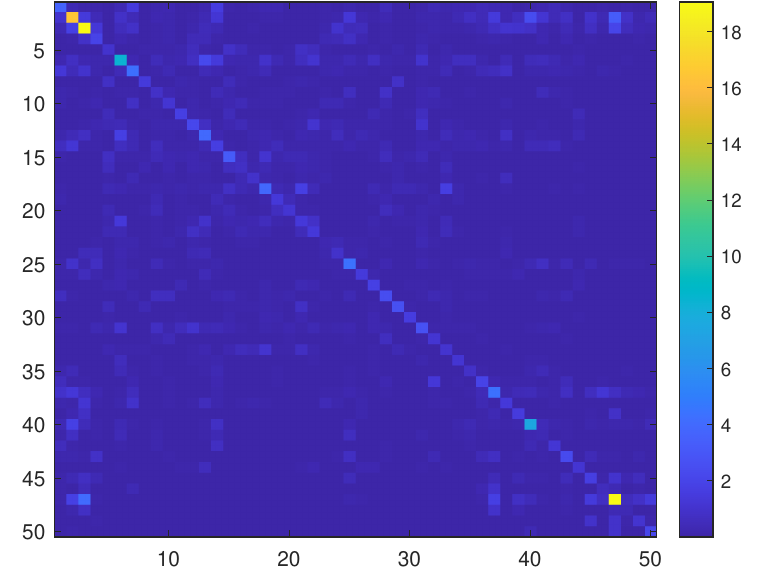}
	\caption{Confusion matrix for the 30-bus block-form lossless case.}
	\label{fig:confusion-matrix}
\end{figure}

\begin{figure}[!t]
	\centering
	\begin{subfigure}[t]{1.0\linewidth}
		\centering\includegraphics[width=1.0\linewidth]{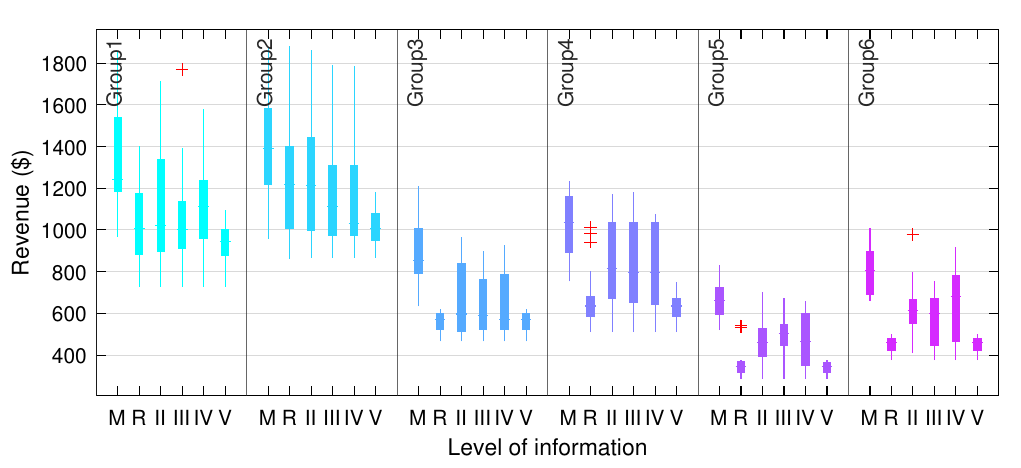}
		\caption{Lossless OPF with block-form bidding curve.}
		\label{subfig:boxplot-case30-stepwise-lossless}
	\end{subfigure}
	\\
	\begin{subfigure}[t]{1.0\linewidth}
		\centering\includegraphics[width=1.0\linewidth]{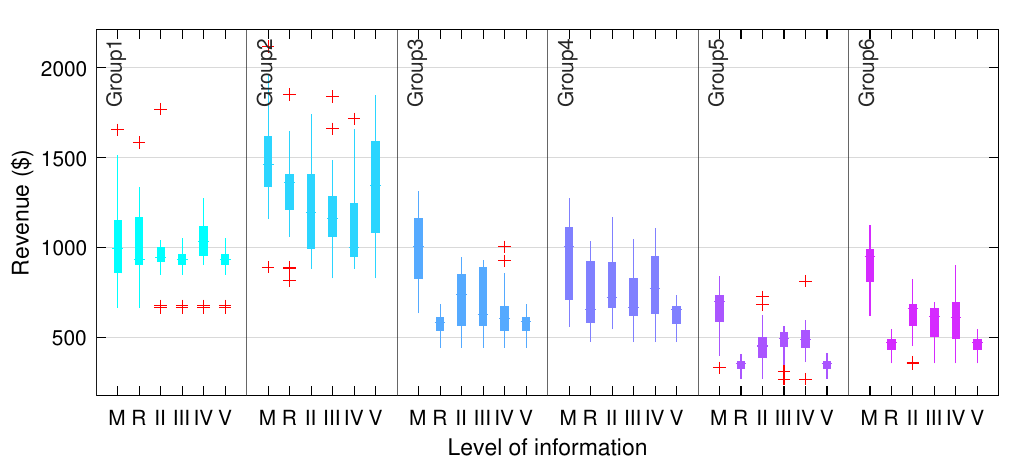}
		\caption{Lossy OPF with block-form bidding curve.}
		\label{subfig:boxplot-case30-stepwise-lossy}
	\end{subfigure}
	\\
	\begin{subfigure}[t]{1.0\linewidth}
		\centering\includegraphics[width=1.0\linewidth]{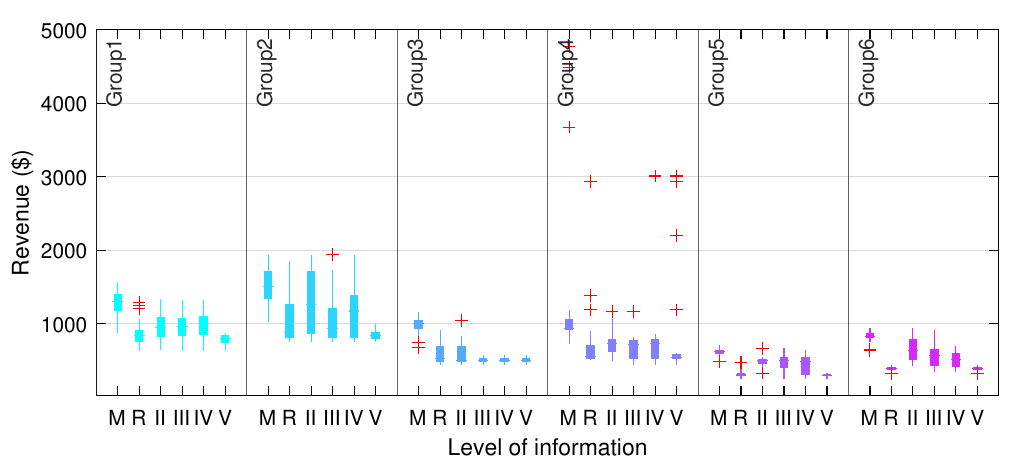}
		\caption{Lossless OPF with quadratic-form bidding curve.}
		\label{subfig:boxplot-case30-quadratic-lossless}
	\end{subfigure}
	\\
	\begin{subfigure}[t]{1.0\linewidth}
		\centering\includegraphics[width=1.0\linewidth]{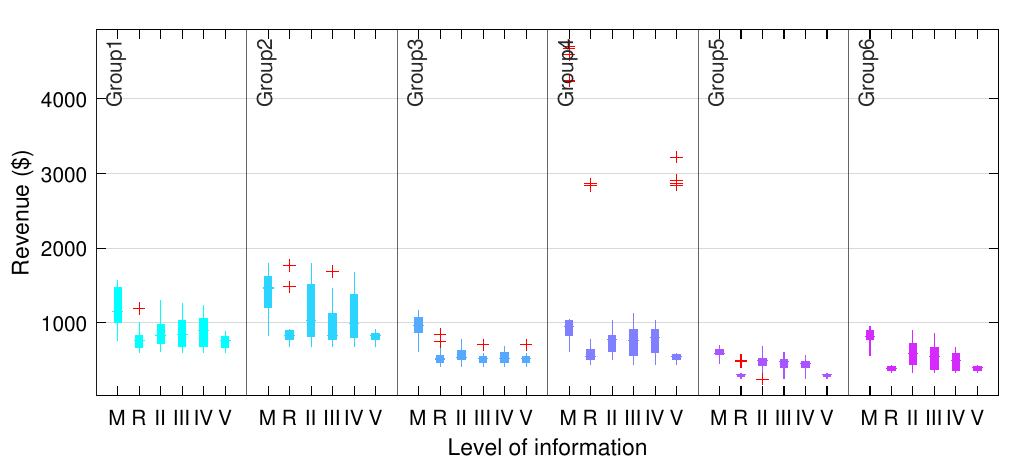}
		\caption{Lossy OPF with quadratic-form bidding curve.}
		\label{subfig:boxplot-case30-quadratic-lossy}
	\end{subfigure}
	\caption*{\small Axis: M denotes the MPEC method; R denotes the RDC-like method.}
	\caption{Boxplots of revenues for different generators in the IEEE 30-bus system. Six groups represent the GenCos owning one of the six strategic generators, respectively.}
	\label{fig:boxplot-case30}
\end{figure}

The performance of SVM in the multi-class classification task is presented. In comparison, two types of dummy classifiers are also used: The first one (Dummy 1) always predicts the most frequent class, and the second one (Dummy 2) is a random guess according to empirical probability. 

Table~\ref{table:classification} shows the classification accuracy under different case settings. The accuracy is defined as the proportion of correctly classified samples in the whole set. The 50 most frequent CRs are chosen while other CRs are considered rare and filtered. As shown in the table, the SVM outperforms the dummy classifiers with much higher accuracies in all case settings. 
For the block form, the CRs become more diversified so that the SVM test accuracy is between 42\% and 62\%. For the quadratic form, the test accuracy is between 66\% and 93\%.
In the 30-bus case, the test accuracy of lossy OPF is slightly higher than the accuracy of lossless OPF. This is because in the lossless case the 50 most frequent CRs cover more data than in the lossy case. Also, the overlap of CRs is not significant in this case.

It is also shown that the classifiers in larger cases perform somehow worse than in the 30-bus cases. This is because the overlaps between CRs become stronger when the power system is more complex. 

The above accuracy is used for deterministic classification problems. 
A well-known metric for probabilistic classification is the confusion matrix, which in the multi-class ($ K $-class) case is defined as follows:
\begin{equation}
\left[\mathrm{Conf}\right]_i = \left[ \sum_{j: y^{(j)} = i} \boldsymbol{p}^{(j)} \right]_i, \ i = 1,2,\cdots,K
\end{equation}
The diagonal dominance level of the confusion matrix represents the performance of the probabilistic classifier. Fig.~\ref{fig:confusion-matrix} shows the confusion matrix for the 30-bus block-form lossless case with a  strong dominance.

\subsection{Bid Results in IEEE 30-Bus System}

Since analyzing the cost of generators is not the focus of this paper, the cost functions for different generators $ h_g(q_g) $ in (\ref{equ:max-revenue}) are set uniformly at a relatively low level:
\begin{equation} \label{equ:cost_function}
h_g(q_g) = \frac{1}{2}\times 0.1 q_g^2 + 5 q_g
\end{equation}

By solving the revenue maximization problem according to different levels of information, 20 time intervals from the 10 most frequent CRs are used to test the performance. The boxplots of the revenues in different cases are plotted in Fig.~\ref{fig:boxplot-case30}.  As shown in the figure, the proposed method outperforms the benchmark using empirical probability (V) in all levels of information. Note that although MPEC (level I) can obtain the theoretically maximum revenue, the big-M method may induce certain loss w.r.t. selection of the $ M $ value. Instead of carefully choosing $ M $ for every single time interval and every GenCo, we fix $ M $ at 1\,000 in this case. 
The RDC-like method (R) performs slightly better than the benchmark using empirical probability (V), because it assumes the CR is known and fixed to the generator. However, the proposed method outperforms the RDC-like method in most of the cases.

Table~\ref{table:case30-revenue} shows the average revenue values of all six generators in the test set for the 30-bus case. The percentage values representing the proportion of revenue obtained under imperfect information to the theoretical revenue under perfect information by MPEC are also given. 
It is shown that the proposed method achieves on average approximately 70\%-76\% of the maximum revenue at information level II-IV, while the benchmark considering only the empirical probability achieves approximately 62\% of the maximum revenue. The absolute improvement of the proposed method is approximately 8\%-14\%.

\begin{table*}[!t]
	\renewcommand{\arraystretch}{1.3}
	\caption{Average Revenue and Percentage Values for the 30-Bus Case}
	\label{table:case30-revenue}
	\centering
	\begin{tabular}{@{\extracolsep{2pt}}cccccccccccccc}
		\toprule
		\multirow[b]{2}{*}{\makecell{Bid curve\\form}} & \multirow[b]{2}{*}{\makecell{OPF\\type}} & \multicolumn{2}{c}{MPEC} & \multicolumn{2}{c}{RDC-like} & \multicolumn{2}{c}{Level II} & \multicolumn{2}{c}{Level III} & \multicolumn{2}{c}{Level IV} & \multicolumn{2}{c}{Level V} \\ \cline{3-4} \cline{5-6} \cline{7-8} \cline{9-10} \cline{11-12} \cline{13-14}
		& & \thead{Revenue\\(\$)} & \thead{Pct.\\(\%)} & \thead{Revenue\\(\$)} & \thead{Pct.\\(\%)} & \thead{Revenue\\(\$)} & \thead{Pct.\\(\%)} & \thead{Revenue\\(\$)} & \thead{Pct.\\(\%)}  & \thead{Revenue\\(\$)} & \thead{Pct.\\(\%)} & \thead{Revenue\\(\$)} & \thead{Pct.\\(\%)} \\ \midrule
		Block & Lossless & 1024.74 & 100 & 718.28 & 70.09 & 820.55 & 80.07 & 793.28 & 77.41 & 802.03 & 78.27 & 655.10 & 63.93 \\
		Block & Lossy & 1000.53 & 100 & 735.66 & 73.53 & 794.35 & 79.39 & 760.09 & 75.97 & 777.51 & 77.71 & 708.62 & 70.82 \\ 
		Quadratic & Lossless & 1067.71 & 100 & 662.70 & 59.71 & 804.99 & 75.39 & 704.71 & 66.00 & 720.82 & 67.51 & 602.85 & 56.46 \\ 
		Quadratic & Lossy & 1121.66 & 100 & 613.36 & 56.43 & 772.16 & 68.84 & 712.73 & 63.54 & 717.38 & 63.96 & 635.12 & 56.62 \\ \hline
		\multicolumn{2}{c}{Average} & 1053.66 & 100 &  682.50 &  64.77  &  798.01 &	75.74	& 742.70 & 70.49 & 754.44	& 71.60 & 650.42 & 61.73 \\ \bottomrule	
	\end{tabular}
\end{table*}

\begin{table*}[!t]
	\renewcommand{\arraystretch}{1.3}
	\caption{Average Revenue and Percentage Values for the 200-Bus Case}
	\label{table:case200-revenue}
	\centering
	\begin{tabular}{@{\extracolsep{2pt}}cccccccccccccc}
		\toprule
		\multirow[b]{2}{*}{\makecell{Bid curve\\form}} & \multirow[b]{2}{*}{\makecell{OPF\\type}} & \multicolumn{2}{c}{MPEC} & \multicolumn{2}{c}{RDC-like} & \multicolumn{2}{c}{Level II} & \multicolumn{2}{c}{Level III} & \multicolumn{2}{c}{Level IV} & \multicolumn{2}{c}{Level V} \\ \cline{3-4} \cline{5-6} \cline{7-8} \cline{9-10} \cline{11-12} \cline{13-14}
		& & \thead{Revenue\\(\$)} & \thead{Pct.\\(\%)} & \thead{Revenue\\(\$)} & \thead{Pct.\\(\%)} & \thead{Revenue\\(\$)} & \thead{Pct.\\(\%)} & \thead{Revenue\\(\$)} & \thead{Pct.\\(\%)}  & \thead{Revenue\\(\$)} & \thead{Pct.\\(\%)} & \thead{Revenue\\(\$)} & \thead{Pct.\\(\%)} \\ \midrule
		Block & Lossless & 3837.30 & 100 & 1793.23 & 46.73 & 2172.02 & 56.6 & 2078.83 & 54.17 & 2121.78 & 55.29 & 1793.23 & 46.73 \\
		Block & Lossy & 3245.22 & 100 & 1711.78 & 52.75 & 1769.69 & 54.53 & 1740.13 & 53.62 & 1835.13 & 56.55 & 1714.98 & 52.85 \\ 
		Quadratic & Lossless & 3900.77 & 100 & 2361.82 & 60.55 & 2643.75 & 67.78 & 2636.15 & 67.58 & 2337.74 & 59.93 & 1994.27 & 51.12  \\ 
		Quadratic & Lossy & 3273.78 & 100 & 2044.77 & 62.46 & 2532.62 & 77.36 & 2486.67 & 75.96 & 2233.58 & 68.23 & 2227.42 & 68.04 \\ \hline
		\multicolumn{2}{c}{Average} & 3564.27 & 100 & 1977.9 & 55.49 & 2279.52 & 63.95 & 2235.45 & 62.72 & 2132.06 & 59.82 & 1932.48 & 54.22		
		\\ \bottomrule	
	\end{tabular}
\end{table*}

\begin{figure}[!t]
	\centering
	\includegraphics[width=0.9\linewidth]{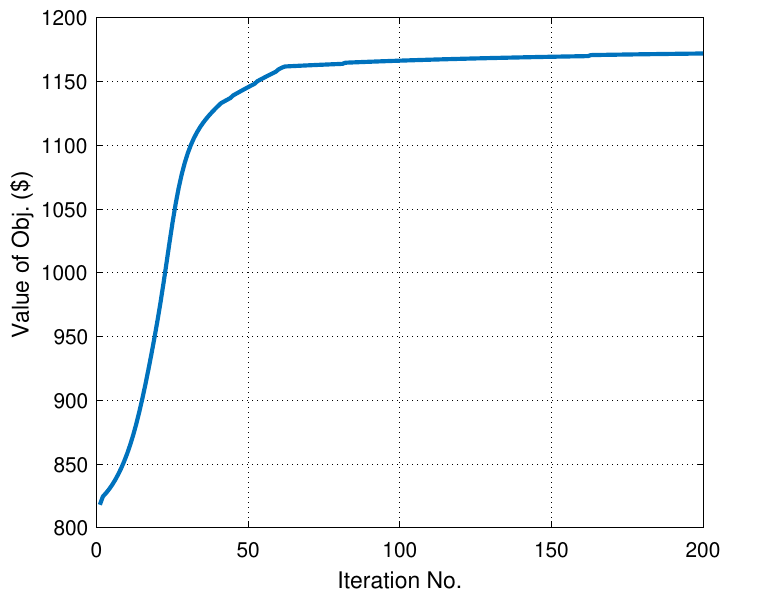}
	\caption{An example on the change of objective function values in gradient descent.}
	\label{fig:iteration-case30}
\end{figure}

Fig.~\ref{fig:iteration-case30} gives an example of the change in objective function values during the gradient descent procedure. Usually, the proposed method converges within 100 iterations. The time consumption for a 200-iteration procedure is approximately 45 seconds.

\subsection{Bid Results in Illinois 200-Bus System}

The cost functions are set the same as (\ref{equ:cost_function}) in the 200-bus system case. Table~\ref{table:case200-revenue} shows the average revenue values of the six strategic generators in the test set for the 200-bus case. The proposed method achieves on average approximately 59\%-65\% of the maximum revenue at information level II-IV, while the benchmark method can only achieve approximately 54\% of the maximum revenue. Thus, the improvement of the proposed method under imperfect information is approximately 5\%-11\%.

Compared with the results in the 30-bus case, the improvement becomes less when the system becomes larger. In the 200-bus case, there are many nonstrategic generators and the market power of each strategic generator is smaller. Thus, information extracted from historical market data decreases, and the performance of the proposed method is therefore negatively affected. 

The time consumption for a 200-iteration procedure in this case is approximately 120 seconds.

\subsection{Bid Results in South Carolina 500-Bus System}

The cost functions are set the same as (\ref{equ:cost_function}). Table~\ref{table:case500-revenue} shows the average revenue values of the two generators at bus \#16 and \#17. The results of level IV are not presented, because the Matpower case file only defines two regions for the 500-bus system. If the zonal load data are used to replace the nodal load data in this case, most  information contained in the nodal load vector will be lost.

The main difference between the 500-bus cases and the smaller cases is that the market becomes more competitive when the system becomes larger and the number of generators increases. Thus, even if MPEC is used to calculate the optimal generation offer under perfect information, it cannot bring a significant extra revenue compared with other methods under imperfect information. The performance of MPEC is slightly worse than some other methods for the lossy OPF and block form of bid curves, since it is difficult to consider the loss factors in the lower level model of MPEC. In fact, MPEC only uses a lossless OPF which cannot fully simulate the lossy OPF. 

Although the market power of the generators are smaller compared with the 30-bus and 200-bus cases, the proposed method with information levels II and III still performs better than the benchmark methods. 

\begin{table*}[!t]
	\renewcommand{\arraystretch}{1.3}
	\caption{Average Revenue and Percentage Values for the 500-Bus Case}
	\label{table:case500-revenue}
	\centering
	\begin{tabular}{@{\extracolsep{2pt}}cccccccccccc}
		\toprule
		\multirow[b]{2}{*}{\makecell{Bid curve\\form}} & \multirow[b]{2}{*}{\makecell{OPF\\type}} & \multicolumn{2}{c}{MPEC} & \multicolumn{2}{c}{RDC-like} & \multicolumn{2}{c}{Level II} & \multicolumn{2}{c}{Level III} & \multicolumn{2}{c}{Level V} \\ \cline{3-4} \cline{5-6} \cline{7-8} \cline{9-10} \cline{11-12} 
		& & \thead{Revenue\\(\$)} & \thead{Pct.\\(\%)} & \thead{Revenue\\(\$)} & \thead{Pct.\\(\%)} & \thead{Revenue\\(\$)} & \thead{Pct.\\(\%)} & \thead{Revenue\\(\$)} & \thead{Pct.\\(\%)}  & \thead{Revenue\\(\$)} & \thead{Pct.\\(\%)} \\ \midrule
		Block & Lossless & 24711.5 & 100 & 23393.2 & 94.67 & 23785.7 & 96.25 & 23891.9 & 96.68 & 22883.2 & 92.60 \\
		Block & Lossy & 25389.2 & 100 & 26378.7 & 103.9 & 26158.1 & 103.0 & 26397.0 & 104.0 & 25166.1 & 99.12 \\ 
		Quadratic & Lossless & 13227.9 & 100 & 11787.2 & 89.11 & 12095.0 & 91.44 & 11885.1 & 89.85 & 11740.3 & 88.75  \\ 
		Quadratic & Lossy & 14564.4 & 100 & 12723.5 & 87.36 & 13083.9 & 89.83 & 13062.2 & 89.69 & 12737.4 & 87.46 \\ \hline
		\multicolumn{2}{c}{Average} & 19473.3 & 100 & 18570.7 & 95.36 & 18780.7 & 96.44 & 18809.1 & 96.60 & 18131.8 & 93.11
		\\ \bottomrule	
	\end{tabular}
\end{table*}

\subsection{Discussion on Scalability to Larger Systems}

Generally, as the power system become larger to thousands of nodes, the electricity market would become more competitive. It is difficult for GenCos to seek extra revenue in a perfect competitive market. From the case study, it is shown that the proposed method could extract useful information from market data and support bidding decision in a moderate-sized power system. If the CR is not complex in a large system, the performance of the proposed method is expected to be similar to the performance in moderated-sized cases. If the CR is too complex, other strategies (\textit{e.g.}, RDC-based methods) are preferred.

\section{Conclusion}
\label{sec:conclusion}
This paper proposes a data-driven pool strategy for price-makers under imperfect information. Based on general OPF models, the critical region and system pattern characteristics of the market clearing model are analyzed using rim-MPLP theory. Then, such characteristics are used for the training of multi-class classifier with historical market data. The parametric functions of optimal dual and primal variables in OPF are fitted with linear regression. The learned classifier and parametric functions are integrated into the price-makers decision framework to form a revenue maximization problem. We then show that such problem can be solved using gradient descent-based methods. The proposed method can consider the impact of GenCo's bidding on both the market binding constraints and market clearing prices. It can be regarded as an extension of the famous MPEC method under imperfect information when essential power system parameters are unavailable. The proposed method can be used for both market participants and market organizers in decision-making, behavior evaluation and simulation.
A case study based on the IEEE 30-bus system, Illinois 200-bus system, and South Carolina 500-bus system demonstrates the performance of the proposed method in small and moderate-sized systems.

Future work includes considering transmission network contingency in the method.


\appendices
\section{Gradient of the Probabilities}
\label{appendix}

Assume an infinitesimal $ \Delta \boldsymbol{X} = [\Delta X_1, \Delta X_2,\cdots ]^\top $ is added to $ \boldsymbol{X} $. For SVM, we have:
\begin{equation} \label{equ:delta-decision}
\Delta f_{ij} = \boldsymbol{w}_{ij}^\top \Delta \boldsymbol{X}, \ i<j
\end{equation}
And when $ i<j $:
\begin{equation}
\begin{aligned}
\Delta r_{ij} & = - \frac{1}{(1+ e^{A_{ij} f_{ij} + B_{ij} })^2 } e^{A_{ij} f_{ij} + B_{ij} } A_{ij} \Delta f_{ij} \\
& = - A_{ij} r_{ij}^2 (1/r_{ij} -1 ) \boldsymbol{w}_{ij}^\top \Delta \boldsymbol{X} \\
& = - A_{ij} r_{ij} (1 - r_{ij} ) \boldsymbol{w}_{ij}^\top \Delta \boldsymbol{X}
\end{aligned}
\end{equation}
When $ i>j $:
\begin{equation}
\begin{aligned}
\Delta r_{ij} & = - \Delta r_{ji} = A_{ji} r_{ji} (1 - r_{ji} ) \boldsymbol{w}_{ji}^\top \Delta \boldsymbol{X} \\
& = A_{ij} r_{ij} (1-r_{ij} ) \boldsymbol{w}_{ij}^\top \Delta \boldsymbol{X}
\end{aligned}
\end{equation}
With slight abuse of notation, $ A_{ij} = A_{ji} $ and $ \boldsymbol{w}_{ij} = \boldsymbol{w}_{ji} $ for $ i\neq j $. 
Then, to calculate $ \Delta \boldsymbol{Q} $, for $ i = j $:
\begin{equation}
\begin{aligned}
\Delta Q_{ij} & = \sum_{s: s\neq i} 2r_{si} \Delta r_{si} \\
& = \left( \sum_{s: s\neq i} A_{si}(-1)^{\mathds{1}(s<i)} 2 r_{si}^2 (1-r_{si}) \boldsymbol{w}_{si}^\top \right) \Delta \boldsymbol{X}
\end{aligned}
\end{equation}
Where $ \mathds{1}(\cdot) $ is the boolean function for the input condition. For $ i < j $:
\begin{equation}
\begin{aligned}
\Delta Q_{ij} & = \Delta (r_{ij}^2 - r_{ij}) \\ 
&= -A_{ij} (2r_{ij} - 1) r_{ij} (1 - r_{ij}) \boldsymbol{w}_{ij}^\top \Delta \boldsymbol{X} \\
\end{aligned}
\end{equation}
Since $ \boldsymbol{Q} $ is symmetric, for $ i >j $:
\begin{equation}
\begin{aligned}
\Delta Q_{ij} & = \Delta Q_{ji} = -A_{ji} (2r_{ji} - 1) r_{ji} (1 - r_{ji}) \boldsymbol{w}_{ji}^\top \Delta \boldsymbol{X} \\
& = A_{ij} (2r_{ij}-1) r_{ij} (1-r_{ij}) \boldsymbol{w}_{ij}^\top \Delta \boldsymbol{X}
\end{aligned}
\end{equation}
Now we have $ \Delta \boldsymbol{Q} $ w.r.t. $ \Delta \boldsymbol{X} $. From (\ref{equ:kkt-pQ}) we have:
\begin{equation}
\begin{gathered}
(\Delta \boldsymbol{Q}) \boldsymbol{p} + \boldsymbol{Q} \Delta \boldsymbol{p} + \boldsymbol{e} \Delta \alpha = 0 \\
\boldsymbol{e}^\top \Delta \boldsymbol{p} = 0 
\end{gathered}
\end{equation}
Since $ \boldsymbol{Q} $ is positive definite:
\begin{equation} \label{equ:delta_p}
\begin{aligned}
& \boldsymbol{e}^\top \boldsymbol{Q}^{-1} \left[ (\Delta \boldsymbol{Q}) \boldsymbol{p} + \boldsymbol{e} \Delta \alpha \right] = 0 \\
\Rightarrow \ & \Delta \alpha = \left( \boldsymbol{e}^\top \boldsymbol{Q}^{-1} \boldsymbol{e} \right)^{-1} \left[ \boldsymbol{e}^\top \boldsymbol{Q}^{-1} (\Delta \boldsymbol{Q}) \boldsymbol{p} \right] \\
\Rightarrow \ & \Delta \boldsymbol{p} = \boldsymbol{Q}^{-1} \left[ \boldsymbol{e} \left( \boldsymbol{e}^{\top} \boldsymbol{Q}^{-1} \boldsymbol{e} \right)^{-1} \boldsymbol{e}^\top \boldsymbol{Q}^{-1} - \boldsymbol{I} \right] (\Delta \boldsymbol{Q}) \boldsymbol{p}
\end{aligned}
\end{equation}

\begin{table*}[!t]
	\centering
	\begin{minipage}{0.75\textwidth}
		\begin{equation} \label{equ:deltaQ}
		\Delta Q_{ij} = \begin{cases}
		\sum_{s: s\neq i} A_{si}(-1)^{\mathds{1}(s<i)} 2 r_{si}^2 (1-r_{si}) w_{si,k} \Delta X_k & \mathrm{if}\ i = j, \\
		A_{ij}(-1)^{\mathds{1}(i<j)}  (2r_{ij} - 1) r_{ij} (1 - r_{ij}) w_{ij,k} \Delta X_k & \mathrm{if}\ i \neq j.
		\end{cases}
		\end{equation}
		\begin{equation} \label{equ:Dk}
		\left[\boldsymbol{D}_k\right]_{ij} = \begin{cases}
		\sum_{s: s\neq i} A_{si}(-1)^{\mathds{1}(s<i)} 2 r_{si}^2 (1-r_{si}) w_{si,k} & \mathrm{if}\ i = j,  \\ 
		A_{ij}(-1)^{\mathds{1}(i<j)}  (2r_{ij} - 1) r_{ij} (1 - r_{ij}) w_{ij,k} & \mathrm{if}\ i \neq j.
		\end{cases}
		\end{equation}
		\medskip
		\hrule
	\end{minipage}
\end{table*}

Let $ \boldsymbol{M} = \boldsymbol{Q}^{-1} \left[ \boldsymbol{e} \left( \boldsymbol{e}^{\top} \boldsymbol{Q}^{-1} \boldsymbol{e} \right)^{-1} \boldsymbol{e}^\top \boldsymbol{Q}^{-1} - \boldsymbol{I} \right] $. To further simplify the gradient of $ \boldsymbol{p} $ w.r.t. $ \boldsymbol{X} $, we focus on the change of $ X_k $ for some index $ k $. When fixing $ X_{k'} $ for all $ k'\neq k $, $ \Delta Q_{ij} $ can be represented as~(\ref{equ:deltaQ}). 
Then, define matrix $ \boldsymbol{D}_k $ as (\ref{equ:Dk}).

Thus, $ \partial \boldsymbol{Q} / \partial X_k  = \boldsymbol{D}_K $. And:
\begin{equation}
\begin{aligned}
& \Delta \boldsymbol{p} = \boldsymbol{M} \boldsymbol{D}_k \boldsymbol{p} \Delta X_k \\
\Rightarrow & \frac{\partial \boldsymbol{p}}{\partial X_k} = \boldsymbol{M}\boldsymbol{D}_k \boldsymbol{p} 
\end{aligned}
\end{equation}
Due to linearity, the gradient is:
\begin{equation}
\boldsymbol{\nabla}_{\boldsymbol{X}} \boldsymbol{p} = \boldsymbol{M} \begin{bmatrix}
\boldsymbol{D}_1 \boldsymbol{p}, \boldsymbol{D}_2 \boldsymbol{p}, \cdots, \boldsymbol{D}_K \boldsymbol{p}
\end{bmatrix}
\end{equation}
By defination, $ \left( \boldsymbol{\nabla}_{\boldsymbol{x}} \boldsymbol{p} \right)_{\mathcal{G}} $ is part of $ \boldsymbol{\nabla}_{\boldsymbol{X}} \boldsymbol{p} $. 

Finally, the only thing that relates to SVM is (\ref{equ:delta-decision}). The above derivation holds for other classfiers as long as $ \Delta f_{ij} $ can be written w.r.t. $ \Delta \boldsymbol{X} $.


%





\ifCLASSOPTIONcaptionsoff
  \newpage
\fi



%
\bibliographystyle{IEEEtran}
\bibliography{ref}

%

\begin{IEEEbiography}[{\includegraphics[width=1in,height=1.25in,clip,keepaspectratio]{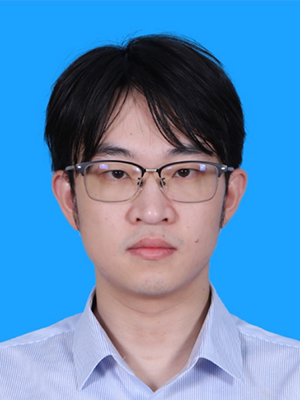}}]{Kedi Zheng} (S'17-M'22) 
	received the B.S. and Ph.D. degrees in electrical engineering from Tsinghua University, Beijing, China, in 2017 and 2022, respectively.
	
	He is currently a Post-Doctoral Researcher with Tsinghua University. His research interests include data analytics in power systems and electricity markets.
\end{IEEEbiography}

\begin{IEEEbiography}[{\includegraphics[width=1in,height=1.25in,clip,keepaspectratio]{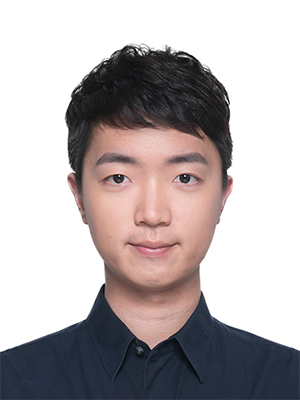}}]{Hongye Guo} (S’15-M’20)
	received the B.S. and Ph.D. degrees in electrical engineering from Tsinghua University, Beijing, China, in 2015 and 2020, respectively.
	
	He is currently a postdoc research fellow at Tsinghua University. His research interests include electricity markets, game theory, energy economics and machine learning.
\end{IEEEbiography}

\begin{IEEEbiography}[{\includegraphics[width=1in,height=1.25in,clip,keepaspectratio]{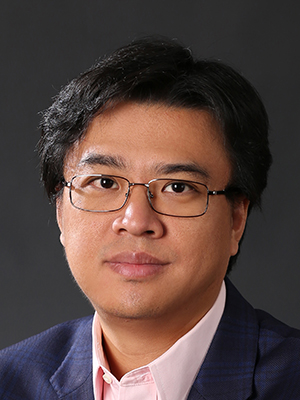}}]{Qixin Chen} (M'10-SM'15) 
	received a Ph.D. degree from the Department of Electrical Engineering at Tsinghua University, Beijing, China, in 2010.
	
	He is currently an Associate Professor at Tsinghua University. His research interests include electricity markets, power system economics and optimization, low-carbon electricity and power generation expansion planning.
\end{IEEEbiography}


%
%
%


\vfill


\end{document}